\documentclass[11pt]{article}
\usepackage{times}
\usepackage{helvet}
\usepackage{courier}
\usepackage{mleftright}

\usepackage{amsmath}
\usepackage{amssymb, amsthm, amsfonts,textcomp}
\usepackage{latexsym}
\usepackage[noend,ruled]{algorithm2e} 

\usepackage{graphicx}
\usepackage{xcolor}
\usepackage{latexsym}
\usepackage{multirow}
\usepackage[]{units}
\usepackage[margin=1in]{geometry}
\usepackage[colorlinks,citecolor=green!60!black,linkcolor=red!60!black]{hyperref}

\newcommand{\mybinom}[2]{%
    \mleft(
    \begin{array}{@{}c@{\,}} #1\\#2 \end{array}
    \mright)}

\allowdisplaybreaks

\newcommand{\reals}{{{\mathbb{R}}}}

\newcommand{\OPT}{{{{\mathrm{OPT}}}}}

\newcommand{\Ev}{{\mathit{Ev}}}

\newcommand{\hit}{{\mathit{hit}}}
\newcommand{\poly}{{\mathrm{poly}}}

\newcommand{\argmax}{{{{\mathrm{argmax}}}}}

\newcommand{\np}{{\mathrm{NP}}}
\newcommand{\fpt}{{\mathrm{FPT}}}

\newcommand{\wone}{{\mathrm{W[1]}}}

\newcommand{\p}{{\mathrm{P}}}
\newcommand{\integers}{{{\mathbb{Z}}}}

\newcommand{\calC}{{{\mathcal{C}}}}
\newcommand{\calA}{{{\mathcal{A}}}}

\newtheorem{example}{Example}
\newtheorem{theorem}{Theorem}

\newtheorem{corollary}{Corollary}
\newtheorem{problem}{Problem}
\newtheorem{definition}{Definition}
\newtheorem{proposition}{Proposition}
\newtheorem{lemma}{Lemma}

\newenvironment{proof-sketch}{\noindent\textit{Proof sketch.}\quad}{\hfill\qed\medskip}

\SetArgSty{textrm}  
\SetAlFnt{\small}
\SetAlCapFnt{\small}
\SetAlCapNameFnt{\small}
\SetAlCapHSkip{0pt}
\IncMargin{-\parindent}

\begin{document}
%
\title{FPT Approximation Schemes for Maximizing Submodular Functions\footnote{The preliminary version of this paper was presented at the 12th Conference on Web and Internet Economics (WINE-2016).}}
\date{}

\author{
Piotr Skowron\\
       {Technische Universit\"at Berlin}\\
       {Berlin, Germany}\\
}

\maketitle

\begin{abstract}
We investigate the existence of approximation algorithms for maximization of submodular functions, that run in a fixed parameter tractable (FPT) time. Given a non-decreasing submodular set function $v\colon 2^X \to \reals$ the goal is to select a subset $S$ of $K$ elements from $X$ such that $v(S)$ is maximized.  We identify three properties of set functions, referred to as $p$-separability properties, and we argue that many real-life problems can be expressed as maximization of submodular, $p$-separable functions, with low values of the parameter $p$. We present FPT approximation schemes for the minimization and maximization variants of the problem, for several parameters that depend on characteristics of the optimized set function, such as $p$ and $K$.
\end{abstract}

\section{Introduction}

We study (exponential-time) approximation algorithms for maximizing non-decreasing submodular set functions. A set function $v\colon 2^X \to \reals$ is submodular if for each two subsets $A \subseteq B \subset X$ and each element $x \in X \setminus B$ it holds that $v(A \cup \{x\}) - v(A) \geq v(B \cup \{x\}) - v(B)$; $v$ is non-decreasing if for each two subsets $A \subseteq B \subset X$ it holds that $v(A) \leq v(B)$. We consider the problem where the goal is to select a subset $S$ of $K$ elements from $X$ such that the value $v(S)$ is maximal.

Maximization of non-decreasing submodular functions is a very general problem that is extensively used in various research areas, from recommendation systems~\cite{budgetSocialChoice, owaWinner}, through voting theory~\cite{sko-fal-sli:j:multiwinner, budgetSocialChoice}, image engineering~\cite{jegelkaBlimes2011, gunhee-iccv-11, NIPS2013_4988}, information retrieval~\cite{conf/nips/YueG11, Lin:2010:MSV:1857999.1858133}, network design~\cite{KrauseGuestrin:uai05infogain, Krause:2008:NSP:1390681.1390689}, clustering~\cite{NIPS2005_2760}, speech recognition~\cite{conf/interspeech/LinB09}, to sparse methods~\cite{Structuredsparsity:2011vl, ICML2011Das_542}. Algorithms for maximization of non-decreasing submodular functions are applicable to other general problems of fundamental significance, such as the \textsc{MaxCover} problem~\cite{chvatal1979setcover, conf/aaai/SkowronF15}. The universal relevance of the problem implies that the existence of good (approximation) algorithms for it is highly desired.

Indeed, the problem has already received a considerable amount of attention in the scientific community. For instance, it is known that the problem is $\np$-hard (on the other hand, Iwata~et~al.~\cite{Iwata:2001:CSP:502090.502096} have shown that minimization of submodular functions is solvable in polynomial time) and that
the greedy algorithm, i.e., the algorithm that starts with the empty set and in each of $K$ consecutive steps adds to the partial solution such an element from $X$ that increases the value of the optimized function most, is an $(1 - \nicefrac{1}{e})$-approximation algorithm for maximization of non-decreasing submodular functions~\cite{submodular}. The same approximation ratio can be achieved for the distributed~\cite{Kumar:2013:FGA:2486159.2486168} and online~\cite{conf/nips/StreeterG08} variants of the problem. Algorithms for maximizing non-monotone submodular functions have been studied by Feige~et~al.~\cite{journals/siamcomp/FeigeMV11}, and the approximability of the problem with additional constraints has been investigated by Calinescu~et~al.~\cite{Calinescu:2007:MSS:1419497.1419517}, Sviridenko~\cite{journals/orl/Sviridenko04}, Lee~et~al.~\cite{Lee:2009:NSM:1536414.1536459}, and Vondr\'{a}k~et~al.~\cite{Vondrak:2011:SFM:1993636.1993740}. For the survey on maximization of submodular functions we refer the reader to the work of Krause and Golovin~\cite{submodularOverview}.

Unfortunately, the approximation guarantees of the greedy algorithm cannot be improved without compromising the efficiency of computation (there are some notable exceptions when the submodular function has a specific structure, for instance when it has low curvature~\cite{journals/dam/ConfortiC84,Sviridenko:2015}). Indeed, the \textsc{MaxCover} problem can be expressed as maximization of a non-decreasing submodular function, yet it is known that under standard complexity assumptions no polynomial-time algorithm can approximate it better than with ratio $(1 - \nicefrac{1}{e})$~\cite{fei:j:cover}. Motivated by this fact, and provoked by the desire to obtain better approximation guarantees, we turn our attention to algorithms that run in super-polynomial time. In our studies we follow the approach of parameterized complexity theory and look for algorithms that run in fixed parameter tractable time (in FPT time), for some natural parameters. To the best of our knowledge, FPT approximation of optimizing submodular functions has not been considered in the literature before.

Parameterized complexity theory aims at investigating how the complexity of a problem depends on the size of different parts of input instances, called parameters. An algorithm runs in FPT time for a parameter $P$ if it solves each instance $I$ of the problem in time $O(f(|P|)\cdot\mathrm{poly}(|I|))$, where $f$ is a computable function. This definition excludes a large class of algorithms, such as the ones with complexity $O(|I|^{|P|})$. From the point of view of parameterized complexity, FPT is seen as the class of easy problems. Intuitively, the complexity of an FPT algorithm consists of two parts: $f(|P|)$, which is relatively low for small values of the parameter, and $\mathrm{poly}(|I|)$ which is relatively low even for larger instances, because of polynomial relation between the computation time and the size of an instance. For details on parameterized complexity theory, we point the reader to appropriate overviews~\cite{dow-fel:b:parameterized,nie:b:invitation-fpt,flu-gro:b:parameterized-complexity,books/sp/CyganFKLMPPS15}.

When referring to running times of the algorithms we assume that the size of the input is $|X|$, and we count each access to the submodular function as an elementary operation (i.e., we assume that the access to the submodular function is given by an oracle). For example, when we say that an algorithm for maximizing submodular functions runs in polynomial-time, then in particular this means that the number of accesses of this algorithm to the submodular functions is of the order of $O(|X|^c)$ for some constant $c$. 

We identify several parameters that we believe are suitable for the analysis of the complexity of the maximization of non-decreasing submodular functions. Perhaps the most natural parameter to consider is the required size of solutions, $K$. Our other parameters depend on characteristics of the optimized set function. Specifically, we define a new property of set functions, called $p$-separability, and provide evidence that $p$ is a natural parameter to consider. We do that in Section~\ref{sec:applications}, by presenting several examples of real-life computational problems that can be expressed as maximization of submodular $p$-separable set functions, where the value of $p$ is small.

Our main contribution is a presentation and an analysis of a few algorithms for the problem. We construct fixed parameter tractable approximation schemes, i.e., collections of algorithms that run in FPT time and that can achieve arbitrarily good approximation ratios. We provide algorithms for two variants of the problem: in the first variant, referred to as the \emph{maximization variant}, the goal is to maximize the value $v(S)$. In the second one, referred to as the \emph{minimization variant}, the goal is to minimize $\big(v(X) - v(S)\big)$. While these two variants of the problem have the same optimal solutions, they are not equivalent in terms of their approximability. Indeed, if there exists a solution $S$ with objectively high value, i.e., if $v(S)$ is close to $v(X)$, then an approximation algorithm for the minimization variant of the problem will be usually superior. For instance, if there exists a solution $S$ such that $v(S) = 0.95 \cdot v(X)$, then a $2$-approximation algorithm for the minimization variant of the problem is guaranteed to return a solution with the value better than $0.9 \cdot v(X)$. On the other hand, a $\nicefrac{1}{2}$-approximation algorithm for the maximization variant of the problem may return in such a case a solution with value $0.475 \cdot v(X)$. Conversely, if the value of an optimal solution is significantly lower than the value of the whole set $X$, then a good approximation algorithm for the maximization variant of the problem will produce solutions of a better quality.

Our algorithms run in FPT time for the parameter $(K, p)$, where $K$ is the size of the solution, and $p$ is the lowest value such that the set function is $p$-separable. To address the case of functions which are not $p$-separable for any reasonable values $p$, we define a weaker form of approximability, referred to as approximation of the \emph{minimization-or-maximization variant}---here, the goal is to find a subset $S$ that is good in one of the previous two metrics. Such algorithms are also desired as they are guaranteed to find good approximation solutions, provided high quality solutions exist (i.e., if values of the optimal solutions are close to $v(X)$). We show that there exists a randomized FPT approximation scheme for minimization-or-maximization variant of the problem for the parameter $\left(K, \frac{\sum_{x \in X}v(\{x\})}{v(X)}\right)$.

We believe that the consequences of our general results are quite significant. In particular, in Section~\ref{sec:applications}, we prove the existence of FPT approximation schemes for some natural problems in the computational social choice, in the matching theory, and in the theoretical computer science.

\section{Notation and Definitions}\label{sec:prelims}

Let $X$ denote the universe set. We consider a set function $v: 2^X \to \reals$ that is non-negative, i.e., such that for each $S \subseteq X$ we have $v(S) \geq 0$. We recall that a set function $v$ is submodular if for each two subsets $A \subseteq B \subset X$ and each element $x \in X \setminus B$ it holds that $v(A \cup \{x\}) - v(A) \geq v(B \cup \{x\}) - v(B)$. There are numerous equivalent conditions characterizing submodular functions---for a survey we refer the reader to the seminal article of Nemhauser~et.~al.~\cite{submodular}. It is easy to see that if the set function $v$ is non-decreasing and submodular, then for each two subsets $A \subseteq B \subset X$ and each element $x \in X$ it holds that $v(A \cup \{x\}) - v(A) \geq v(B \cup \{x\}) - v(B)$ (here, we do not have to assume that $x \in X \setminus B$).

Below, we define a new class of properties of set functions that we call $p$-separability.

\begin{definition}[$p$-separable set function]\label{def:separability}
A submodular set function $v: 2^X \to \reals$ is: 
\begin{enumerate}
\item \emph{$p$-superseparable}, if for each $S \subseteq X$ we have:
\begin{align}\label{eq:pSeparableCondition}
\sum_{x \in X} \big(v(S \cup \{x\}) - v(S)\big) \geq \sum_{x \in X}v(\{x\}) -  p \cdot v(S) \text{,}
\end{align}
\item \emph{$p$-subseparable}, if for each $S \subseteq X$ we have:
\begin{align}\label{eq:pSeparableCondition3}
\sum_{x \in X} \big(v(S \cup \{x\}) - v(S)\big) \leq p \cdot v(X) -  p \cdot v(S) \text{.}
\end{align}
\item \emph{rev-$p$-subseparable}, if for each $S \subseteq X$ we have:
\begin{align}\label{eq:pSeparableCondition2}
\sum_{x \in X} \big(v(S \cup \{x\}) - v(S)\big) \geq p \cdot v(X) -  p \cdot v(S) \text{.}
\end{align}
\end{enumerate}
\end{definition}

We will show that for $p$-superseparable and $p$-subseparable functions, a smaller value of the parameter $p$ makes the problem easier to approximate. In contrast, for rev-$p$-subseparability, a larger value of $p$ is better.\footnote{In fact, for all our results to hold, in Definition~\ref{def:separability} we do not have to consider all the subsets $S$ of $X$ but only those with the size at most equal to $K$.} 

Each of the three conditions in Definition~\ref{def:separability} imposes some bound on the expression $\sum_{x \in X} \big(v(S \cup \{x\}) - v(S)\big)$. This expression can be seen as the sum of marginal gains of function $v$ at point $S$. If we write $|\nabla v(S)| = \sum_{x \in X} \big(v(S \cup \{x\}) - v(S)\big)$, and additionally use a common assumption that $v(\emptyset) = 0$ then the condition for the $p$-superseparability can be rewritten as $\frac{|\nabla v(\emptyset)| - |\nabla v(S)|}{v(S)} \leq p$. Thus, $p$-superseparability imposes a specific constraint on how the norm of the gradient of $v$ changes relatively to the value of $v$. Similarly, in this language the condition for the $p$-subseparability can be rewritten as $\frac{|\nabla v(S)|}{v(X) - v(S)} \leq p$.

In particular, observe that a nonincreasing and submodular function $v$ is $p$-superseparable if $\sum_{x \in X}v(\{x\}) \leq  p \cdot \min_{x\colon v(\{x\})>0} v(\{x\})$. Indeed, fix $S \subseteq X$ and consider two cases. If $v(S) = 0$ then Inequality~\eqref{eq:pSeparableCondition} is clearly satisfied as function $v$ is nondecreasing. If $v(S) > 0$, then by submodularity of $v$, there exists $y \in S$ such that $v(\{y\}) > 0$. Thus, by monotonicity of $v$ we have that $p \cdot v(S) \geq p \cdot \min_{x\colon v(\{x\})>0} v(\{x\})$, and so by our assumption $p \cdot v(S) \geq \sum_{x \in X}v(\{x\}$.
In this case, the right-hand side of Inequality~\eqref{eq:pSeparableCondition} is negative or equal to zero. Since function $v$ is nonincreasing we infer that the left-hand side of Inequality~\eqref{eq:pSeparableCondition} is non-negative, and consequently that the Inequality~\eqref{eq:pSeparableCondition} holds. Also, it is easy to see that each monotone and submodular function is $|X|$-superseparable and $|X|$-subseparable. Yet, in Section~\ref{sec:applications} we show that the value of parameter $p$ in many natural problems is significantly lower than $|X|$. 

Further, observe that a linear combination with positive coefficients of $p$-superseparable functions is $p$-superseparable. The same comment applies to $p$-subseparability and rev-$p$-subseparability. As we will see in Section~\ref{sec:applications}, this observation is very helpful in proving that certain set functions are $p$-separable.

Example~\ref{ex:maxCover} below illustrates how a variant of \textsc{MaxCover}, a fundamental problem in the theoretical computer science, can be expressed as a maximization of a non-negative, nondecreasing, $p$-separable, submodular function.

\begin{example}\label{ex:maxCover}
In the \textsc{MaxWeightCover} problem, we are given a set of $n$ elements, $N = \{e_1, e_2, \ldots e_n\}$, and a collection of $m$ subsets of $N$, $X = \{S_1, \ldots, S_m\}$. Each element $e_i$ has a weight $w_i$, and the goal is to find a subcollection of $K$ subsets from $X$, $\calC$, which maximizes the total weight of the covered elements: $\sum_{i: i \in \bigcup \calC}w_i$.
A \emph{frequency} of an element $e_i$ is defined as the number of sets that contain $e_i$. 

We will show that the \textsc{MaxWeightCover} problem with the frequency of elements upper-bounded by $p$ can be expressed as the maximization of a non-negative, nondecreasing submodular function which is (i) $p$-superseparable, and (ii) $p$-subseparable. For each set $\calC \subseteq X$ we define $v(\calC)$ as the total weight of elements covered by the sets from $\calC$. Such defined $v$ is non-negative and submodular.
Since the weighted sum of $p$-separable set functions is also $p$-separable, it is sufficient to show $p$-separability of a function $u_i$ which returns $1$ for collections of sets that cover $e_i$, and $0$ for the remaining ones. Since the frequency of the elements is bounded by $p$, then $\sum_{S \in X} u_i(\{S\})$, which is the number of sets that cover $e_i$, is also bounded by $p$. Further, $p\cdot \min_{S\colon u_i(\{S\})>0} u_i(\{S\}) = p$, thus the right-hand side of Inequality~\eqref{eq:pSeparableCondition} is negative or equal to zero: this shows that $u_i$ is $p$-superseparable. 
To show that $u_i$ is also $p$-subseparable, let us fix a collection of sets $\calC \subseteq X$ and let us consider two cases. If $e_i$ is covered by $\calC$, then $\sum_{S \in X} \big(u_i(\calC \cup \{S\}) - u_i(\calC)\big)$ is equal to 0. The right-hand side of Inequality~\eqref{eq:pSeparableCondition3} is non-negative, thus the condition for $p$-subseparability holds. If $e_i$ is not covered by $\calC$, then $\sum_{S \in X} \big(u_i(\calC \cup \{S\}) - u_i(\calC)\big)$ is equal to the number of sets that cover $e_i$, thus it is upper bounded by $p$. The right-hand side of Inequality~\eqref{eq:pSeparableCondition3} is equal to $p \cdot v(X) -  p \cdot v(\calC) = p \cdot 1 - p \cdot 0 = p$. Thus the condition for $p$-subseparability also holds.

Similarly, it can be shown that the \textsc{MaxWeightCover} problem with the frequency of elements lower-bounded by $p$ can be expressed as the maximization of a non-negative, nondecreasing submodular function which is rev-$p$-subseparable.
\end{example}

For a better intuition on the the above definitions, we refer the reader to Section~\ref{sec:applications} where we present several further examples of natural problems which can be expressed as optimization of super/sub-separable functions for low values of the parameter $p$ or rev-subseparable functions for high values of $p$.

A different parameter of submodular functions that has been studied in the literature is the curvature. We say that a function $v$ has curvature $c$, $0 \leq c \leq 1$, if for each $S \subset X$ and \mbox{$x \in X \setminus S$} it holds that \mbox{$v(S \cup \{x\}) - v(S) \geq (1 - c)v(\{x\})$}. It is known that the greedy algorithm achieves approximation ratio of $\frac{1 - e^{-c}}{c}$ for the problem of maximizing a submodular function with curvature $c$~\cite{journals/dam/ConfortiC84}. Recently, Sviridenko~et~al.~\cite{Sviridenko:2015} have improved this result by designing a $(1 - \nicefrac{c}{e})$-approximation algorithm for the problem.
The curvature is not comparable to $p$-separability: there are $p$-separable functions for very good values of the parameter $p$, which at the same time have curvature equal to one. For instance, consider the submodular function $v$ for \textsc{MaxWeightCover} from Example~\ref{ex:maxCover}. If the set of elements $N$ can be covered with a subcollection $\calC$, $\calC \neq X$, then for $S \in X \setminus \calC$ it holds that $v(\{S\}) > 0$ and $v(\calC \cup \{S\}) - v(\calC) = v(X) - v(X) = 0$, so the optimized function has curvature equal to one.

In this paper we investigate the problem of selecting $K$ elements from $X$ that altogether maximize the value of the set function $v$.

\begin{problem}[\textsc{BestKSubset}]
For a set of elements $X$, a polynomially computable set function $v: 2^X \to \reals$, and an integer $K$, the solution to the \textsc{BestKSubset} problem is a set $S \subseteq X$ with $|S| \leq K$, for which $v(S)$ is maximal.
\end{problem}

We are specifically interested in finding approximation algorithms for the \textsc{BestKSubset} problem. We focus on approximating two metrics: (i) the value $v(S)$ in the maximization variant of the problem, and (ii) the value $\big(v(X) - v(S)\big)$ in its minimization variant.

\begin{definition}[Approximation algorithms]\label{def:approximation}
Let $S^*$ denote an optimal solution for \textsc{BestKSubset}:
\begin{enumerate}
\item
Fix $\alpha$, $0 < \alpha < 1$. $\calA$ is an \emph{$\alpha$-approximation algorithm for the maximization variant} of \textsc{BestKSubset}, if for each instance $I$ of \textsc{BestKSubset} it returns a set $S$ such that $v(S) \geq \alpha v(S^*)$.
\item
Fix $\alpha$, $\alpha > 1$. $\calA$ is an \emph{$\alpha$-approximation algorithm for the minimization variant} of \textsc{BestKSubset}, if for each instance $I$ of \textsc{BestKSubset} it returns a set $S$ such that $\big(v(X) - v(S)\big) \leq \alpha \big(v(X) - v(S^*)\big)$.
\item
Fix $\alpha$, $\alpha > 1$. $\calA$ is an \emph{$\alpha$-approximation algorithm for the minimization-or-maximization variant} of \textsc{BestKSubset}, if for each instance $I$ of \textsc{BestKSubset} it returns a set $S$ such that $v(S) \geq \frac{1}{\alpha} v(S^*)$ or $\big(v(X) - v(S)\big) \leq \alpha \big(v(X) - v(S^*)\big)$.
\end{enumerate}
\end{definition}

The definition of an approximation algorithm for minimization-or-maximization variant of \textsc{BestKSubset} requires some additional comment: this definition guarantees that the algorithm finds a good solution provided a high quality solution exists. In other words, if there exists an optimal solution $S^*$ such that the value $\big(v(X) - v(S^*)\big)$ is low compared to $v(S^*)$, then a good approximation solution for the minimization variant of the problem is also a good solution for its maximization variant. For some parameters we present good approximation algorithms for the minimization-or-maximization variant of \textsc{BestKSubset}, even though we do not have as good algorithms neither for the minimization nor maximization variants of the problem.


We are specifically interested in FPT approximation schemes. A collection of algorithms $\calA$ is an FPT approximation scheme for a parameter $P$, if for each constant $\alpha$ there exists an $\alpha$-approximation algorithm in $\calA$ that runs in an FPT time for the parameter $P$. 

\section{Algorithms for Maximizing $\mathbf{p}$-separable Submodular Functions}\label{sec:algorithms}

In this section we present our approximation algorithms for the three variants of the \textsc{BestKSubset} problem, formally stated in Definition~\ref{def:approximation}. Our methods are inspired by the algorithms of Skowron and Faliszewski~\cite{conf/aaai/SkowronF15} for the \textsc{MaxCover} problem. We extend these algorithms to be applicable to the problem of maximizing more general submodular functions.

We start by presenting an FPT approximation scheme for the maximization variant of \textsc{BestKSubset} for submodular $p$-superseparable set functions. The algorithm, formally defined as Algorithm~\ref{alg:submod}, gets as an input an instance of the problem and the required approximation ratio, $\alpha$. It proceeds in two steps: first, it restricts the universe set by selecting a certain number of elements from $X$ with the highest values of the set function $v$. Second, it takes the set $\mathcal{A}$ of elements that were selected in the first step, computes the value of the set function for all $K$-element subsets of $\mathcal{A}$, and returns a subset with the highest value.

Algorithm~\ref{alg:submod} is an FPT approximation scheme for the maximization variant of the problem for the parameter $(K, p)$. Before we prove this fact, however, we note that under standard complexity theoretic assumptions, there exists no FPT exact algorithm for the problem. There even exists no FPT exact algorithm for the parameter $K$ if $p$ is a constant. This follows from our observation in Example~\ref{ex:maxCover}, where we show that the \textsc{MaxCover} problem with frequencies bounded by $p$ can be expressed as maximization of a non-negative, non-decreasing, submodular, $p$-superseparable set function, and from the fact that the \textsc{MaxCover} problem with frequencies bounded by a constant, for the parameter $K$ belongs to the complexity class $\wone$~\cite{conf/aaai/SkowronF15}, and it is unlikely that $\wone \subseteq \fpt$. Further, even for $p = 2$ there exists no polynomial-time approximation scheme (PTAS) for the problem. This is because the \textsc{MaxCover} problem with the frequencies of the elements equal to two is a generalization of the \textsc{MaxVertexCover} problem, for which there is no PTAS unless $\p= \np$ (see, e.g., the work of Croce and Paschos~\cite{cro-pas:j:cover} for a discussion on the approximability of the \textsc{MaxVertexCover} problem). 

\begin{theorem}\label{thm:pSubmodularMaximization}
For each non-negative, non-decreasing, submodular, and $p$-superseparable set function $v: 2^X \to \reals$ and for each $0 \leq \alpha < 1$, Algorithm~\ref{alg:submod} outputs an $\alpha$-approximate solution for the maximization variant of \textsc{BestKSubset}, in time $\poly(n,m) \cdot \mybinom{\frac{p K}{(1 - \alpha)} + K}{K}$.
\end{theorem}

\SetKwInput{KwParameters}{Parameters}
\begin{algorithm}[t]
  \small
\KwParameters{\\$\hspace{3pt}$ $X$ --- the set of elements.\\
         $\hspace{3pt}$ $v$ --- the submodular function $v: 2^X \to \reals$ that is $p$-superseparable.\\
         $\hspace{3pt}$ $\alpha$ --- the required approximation ratio of the algorithm. \\}
  \vspace{2mm}
  \SetAlCapFnt{\small}
  $\mathcal{A} \leftarrow \lceil \frac{p K}{(1 - \alpha)} + K \rceil$ elements $x$ from $X$ with highest values $v(\{x\})$ \;
  \Return{$K$-element subset of $\mathcal{A}$ with the highest value of $v$} \;
  \caption{\small An algorithm for the \textsc{BestKSubset} problem for non-negative, non-decreasing, submodular, and $p$-superseparable set functions.}
  \label{alg:submod}
\end{algorithm}

\begin{proof}
Consider an input instance $I$ of the \textsc{BestKSubset} problem. Let $S$ and $S^{*}$ be, respectively, the solution returned by Algorithm~\ref{alg:submod} and some optimal solution. We set $\OPT = v(S^{*})$ as the value of an optimal solution.

We will show that $v(S) \geq \alpha\OPT$. Naturally, the value $v(S)$ might be lower than $v(S^{*})$. This might happen because $\calA$, the set of the elements considered by the algorithm in its second step, might not contain some elements from $S^*$. We will show that $\ell = |S^{*} \setminus \calA|$ elements from $S^{*} \setminus \calA$ might be replaced by some elements from  $\calA$ which are not present in $S^*$, in a way that decreases the value of $S^{*}$ by at most a small fraction. After such replacement, we will end up with the set containing the elements from $\calA$ only. From this we will infer that the value of the best solution in $\calA$ is lower than the value of an optimal solution by at most a small factor.

Let us order the elements from $S^{*} \setminus \calA$ in some arbitrary way, and let us use the notation $S^{*} \setminus \calA = \{x_1, \ldots, x_\ell\}$.
We will replace the elements $\{x_1, \ldots, x_\ell\}$ with the elements $\{x'_1, \ldots, x'_\ell\}$ from $\calA \setminus S^{*}$ (we will define these elements later), one by one, in $\ell$ consecutive steps.
Let $S_i$ denote the set that we get after replacing elements $x_1, \ldots, x_i$ with $x_1', \ldots, x_i'$, that is let $S_i = (S^{*} \setminus \{x_1, \dots, x_{i}\}) \cup \{x_1', \dots, x_{i}'\}$. In particular, $S_{\ell} \subseteq \calA$. Thus, in the $i$-th step we will replace $x_i$ with $x_i'$ in set $S_{i-1}$.

The elements $x'_1, \ldots, x'_\ell$ are defined by induction, in the following way. Assume that we have already found elements $x'_1, \ldots, x'_{i-1}$ (for $i=1$ it means we have not yet found any element, i.e., that we are looking for the first element in the sequence). We define $x'_i$ to be an element from $\calA \setminus S_{i-1}$ that maximizes $v\Big((S_{i-1} \setminus \{x_{i}\}) \cup \{x_i'\}\Big)$, that is  
when selecting $x'_i$ we aim at maximizing $v(S_i)$.

It may happen that after replacing $x_i$ with $x'_i$, the value of the function $v$ for the new set decreases. Let $\Delta_i$ denote the value of such decrease (or increase if the algorithm were lucky---in such case $\Delta_i$ would be negative):
\begin{align*}
\Delta_i = v(S_{i-1}) - v(S_i) \textrm{.}
\end{align*}
  By the construction of the set $\calA$ and the fact that $x_i \notin \calA$, for every $y \in \calA \setminus S_{i-1}$ we have that $v(\{x_i\}) \leq v(\{y\})$.
  By the way we choose the element $x_i'$, we know that for every $y \in \calA \setminus S_{i-1}$, we have:
\begin{align*}
\Delta_i \leq v(S_{i-1}) - v\Big((S_{i-1} \setminus \{x_{i}\}) \cup \{y\}\Big) \textrm{.}
\end{align*}
Using submodularity and after reformulation we get:
\begin{align*}
\Delta_i &\leq v\Big(S_{i-1} \setminus \{x_i\}\Big) + v\Big(\{x_i\}\Big) - v(\emptyset) - v\Big((S_{i-1} \setminus \{x_{i}\}) \cup \{y\}\Big) \\
         & \leq v\Big(S_{i-1} \setminus \{x_i\}\Big) - v\Big((S_{i-1} \setminus \{x_{i}\}) \cup \{y\}\Big) + v\Big(\{y\}\Big) - v(\emptyset) \textrm{.}
\end{align*}
For any $y \in X$ (in particular for $y \notin \calA \setminus S_{i-1}$), by submodularity and monotonicity, we have that:
\begin{align*}
0 \leq v\Big(S_{i-1} \setminus \{x_{i}\}\Big) + v\Big(\{y\}\Big) - v(\emptyset) - v\Big((S_{i-1} \setminus \{x_{i}\}) \cup \{y\}\Big) \textrm{.}
\end{align*}
Since the set function is non-negative, the inequalities above will still hold if we skip the fragment $v(\emptyset)$. Consequently, since the set function is $p$-superseparable, we get:
\begin{align*}
(|\mathcal{A}| - K)\Delta_i & \leq |\calA \setminus S_{i-1}|\Delta_i = \sum_{y \in \calA \setminus S_{i-1}} \Delta_i +  \sum_{y \in X \setminus (\calA \setminus S_{i-1})} 0 \\
                            & \leq \sum_{y \in X} \bigg( v\Big(S_{i-1} \setminus \{x_{i}\}\Big) + v\Big(\{y\}\Big) - v\Big((S_{i-1} \setminus \{x_{i}\}) \cup \{y\}\Big)\bigg) \\
                            & \leq p\cdot v\Big(S_{i-1} \setminus \{x_{i}\}\Big) \leq p\OPT \textrm{.}
\end{align*}
Which leads to:
\begin{align*}
\Delta_i \leq \frac{p\OPT}{|\mathcal{A}| - K} \leq \frac{\OPT p(1-\alpha)}{pK} = \frac{\OPT (1-\alpha)}{K} \textrm{.}
\end{align*}
Since $\ell \leq K$, we conclude that:
\begin{align*}
\sum_{i=1}^\ell \Delta_i \leq  (1 - \alpha)\OPT.
\end{align*}
That is, after replacing the elements from $S^*$ that do not appear in $\calA$ with sets from $\calA$, the optimal value decreases by at most $(1-\alpha)\OPT$. This means that there are $K$ elements in $\calA$ for which the function $v$ achieves the value equal to at least $\alpha\OPT$. Since the algorithm tries all size-$K$ subsets of $\calA$, it finds a solution with such a value.
\end{proof}

Now, we move to optimizing $p$-subseparable set functions. We know that the simple greedy algorithm (Algorithm~\ref{alg:greedy}) achieves approximation ratio of $1 - \nicefrac{1}{e}$ for the maximization variant of the \textsc{BestKSubset} problem~\cite{submodular}. Below, we show that if the optimized set function is rev-$p$-subseparable then for some values of the parameter $p$ the analysis of the approximation guarantees of this simple greedy algorithm can be further improved. We note that in this case we do not even require the set function to be submodular.

\begin{algorithm}[t]
    \small
    \SetAlCapFnt{\small}
    \KwParameters{\\$\hspace{3pt}$ $X$ --- the set of elements.\\
         $\hspace{3pt}$ $v$ --- the submodular function $v: 2^X \to \reals$ that is rev-$p$-subseparable.\\}
    \vspace{2mm}

    $S = \{\}$\;
    \For{$i\leftarrow 1$ \KwTo $K$}{
       $S \leftarrow S \cup \Big\{ \argmax_{x \in X}\Big(v(S \cup \{x\} - v(S))\Big) \Big\}$
    }
    \Return{C}
    \caption{\small An algorithm for the \textsc{BestKSubset} problem for non-negative, non-decreasing, rev-$p$-subseparable set functions.}
    \label{alg:greedy}
\end{algorithm}

\begin{proposition}\label{theorem:greedy}
The greedy algorithm (Algorithm~\ref{alg:greedy}) is a polynomial-time $\left(1 - e^{-\frac{pK}{|X|}}\right)$-approximation algorithm for the maximization variant of \textsc{BestKSubset} problem with a non-negative, non-decreasing, rev-$p$-subseparable set function.
\end{proposition}
\begin{proof}
The algorithm clearly runs in polynomial time and so we focus only on proving its approximation ratio.
  
We prove by induction that for each $i$, $0 \leq i \leq K$, after the $i$'th iteration of the greedy algorithm's ``for'' loop, the value $\Big(v(X) - v(S)\Big)$ is at most equal to $\Big(v(X) - v(\emptyset)\Big)\Big(1 - \frac{p}{|X|}\Big)^{i}$. Naturally, the assumption is true for $i=0$. Suppose that the inductive assumption holds for some $(i-1)$, $1 \leq i < K$. Let $S$ be the partial solution at the beginning of the $i$-th iteration of the algorithm's ``for'' loop  and let $x_b$ be the element added to the partial solution in this iteration. Since the set function $v$ is rev-$p$-subseparable, it holds that:
\begin{align*}
\sum_{x \in X} \Big(v(S \cup \{x\}) - v(S)\Big) \geq p \cdot v(X) -  p \cdot v(S) \text{.}
\end{align*}
From the pigeonhole principle we get that:
\begin{align*}
\Big(v(S \cup \{x_b\}) - v(S)\Big) \geq \frac{p}{|X|} \Big(v(X) - v(S)\Big) \text{.}
\end{align*}
Thus, we get that:
\begin{align*}
v(X) - v(S \cup \{x_b\}) & \leq v(X) - v(S) + v(S) - v(S \cup \{x_b\}) \\
                         & \leq v(X) - v(S) - \frac{p}{|X|} \Big(v(X) - v(S)\Big) \\
                         & = \Big(v(X) - v(S)\Big) \Big(1 - \frac{p}{|X|}\Big) \leq \Big(v(X) - v(\emptyset)\Big)\Big(1 - \frac{p}{|X|}\Big)^{i} \textrm{.}
\end{align*}
Let $C$ be the solution returned by Algorithm~\ref{alg:greedy}.
\begin{align*}
v(X) - v(C) \leq \Big(v(X) - v(\emptyset)\Big)\Big(1 - \frac{p}{|X|}\Big)^{\frac{|X|}{p}\cdot \frac{pK}{|X|}} \leq \Big(v(X) - v(\emptyset)\Big)e^{-\frac{pK}{|X|}} \textrm{.}
\end{align*}
Since the function $v$ is monotonic, and thus $\OPT \leq v(X)$, and since $v$ is non-negative, and thus $v(\emptyset) \geq 0$, we have that:
\begin{align*}
 v(C) \geq v(X) - \Big(v(X) - v(\emptyset)\Big)e^{-\frac{pK}{|X|}} \geq \OPT\Big(1 - e^{-\frac{pK}{|X|}}\Big) \textrm{.}
\end{align*}
This completes the proof.
\end{proof} 

If we additionally assume that the set function is submodular, then naturally, the standard approximation ratio of $1 - \nicefrac{1}{e}$ of the greedy algorithm for maximizing submodular functions still applies and we can strengthen approximation guarantee from Proposition~\ref{theorem:greedy} to $\Big(1 - e^{-\max\left(\frac{pK}{|X|},1\right)}\Big)$. Proposition~\ref{theorem:greedy} has interesting implication: if we restrict our problem to the class of instances for which $\frac{p}{|X|}$ is lower-bounded by some constant, then there exists a polynomial time approximation scheme (PTAS) for such a restricted problem. This observation is interesting, because for some real-life problems it is more natural to express the value of the parameter $p$ as the fraction of the size of the universe set $X$ rather than as the absolute value.

\begin{corollary}\label{cor:ptas}
Let $\gamma \in \reals$ be a constant. There exists a polynomial time approximation scheme for the maximization variant of the \textsc{BestKSubset} problem with non-negative, non-decreasing, and rev-$(\gamma|X|)$-subseparable set function.
\end{corollary}
\begin{proof}
For each $\epsilon$, $0 < \epsilon < 1$, there exists such a constant $c$ that for $K > c$ we have $\big(1 - e^{-\gamma K}\big) \geq 1 - \epsilon$. For $K > c$ we run the greedy algorithm (and the approximation ratio follows from Proposition~\ref{theorem:greedy}), and for $K \leq c$, we invoke a brute-force algorithm that tries all $K$-element subsets of $X$.
\end{proof}

Finally, we consider the minimization variant of \textsc{BestKSubset} for the case of \emph{$p$-subseparable} submodular set functions. In Algorithm~\ref{alg:pSeparableMin} we present a randomized algorithm for the problem: the algorithm performs several independent runs. Each run, in Algorithm~\ref{alg:pSeparableMin} described by the \texttt{SingleRun} procedure, builds the solution by selecting random elements in $K$ consecutive steps. In each step, an element $x$ is selected with the probability proportional to the marginal increase of the value of the set function caused by adding $x$ to the partial solution. Theorem~\ref{thm:pSeparableMin} below shows that if we repeat such a procedure a sufficient number of times, we are very likely to find a solution with the required approximation ratio. Before we state Theorem~\ref{thm:pSeparableMin} we will prove a more technical lemma which will appear useful in proving the theorem, but also in our further analysis. 

\begin{algorithm}[t!]
   \small
   \SetKwInput{KwParameters}{Parameters}
   \SetKwFunction{SingleRun}{SingleRun}
   \SetKwFunction{Main}{Main}
   \SetKwBlock{Block}
   \SetAlCapFnt{\footnotesize}
   \KwParameters{\\
          $\hspace{3pt}$ $X$ --- the set of elements.\\
          $\hspace{3pt}$ $v$ --- the submodular function $v: 2^X \to \reals$ that is $p$-subseparable.\\

          $\hspace{3pt}$ $\alpha$ --- the required approximation ratio of the algorithm \\
          $\hspace{3pt}$ $\epsilon$ --- the allowed probability of achieving worse than $\alpha$ approximation ratio}
\vspace{1mm}	
    \SingleRun{}:
	\Block{
                $S \leftarrow \emptyset$\; 
		\For{$i \leftarrow 0$ \KwTo $K$}
		{
			$x_r \leftarrow$ randomly select an element from $X \setminus S$\\
                        \hspace{0.8cm} with probability of selecting $x$ proportional to $v(S \cup \{x\}) - v(S)$  \;
			$S \leftarrow S \cup \{x_r\}$\;
		}
		\Return{$S$}\;
	}
	\vspace{-2mm}\hspace{2mm} \\
	\Main{}:
        run \SingleRun{} for $\left\lceil (-\ln \epsilon) / \left(\frac{\alpha - 1}{p\alpha}\right)^K \right\rceil$ times; return the best solution\;

     \caption{\small An algorithm for the minimization variant of the \textsc{BestKSubset} problem with a non-negative, non-decreasing, submodular, and $p$-subseparable set function.}
   \label{alg:pSeparableMin}
\end{algorithm}

\begin{lemma}\label{lem:algorithm2}
Consider a single iteration of the ``for'' loop within the function \texttt{SingleRun} in Algorithm~\ref{alg:pSeparableMin}. Let $S^*$ be some optimal solution for $I$, let $S$ be the value of the partial solution at the beginning of this iteration and let $p_\hit$ denote the probability that in this iteration any element from $S^*$ is added to the partial solution (thus, using notation from Algorithm~\ref{alg:pSeparableMin}, $p_\hit$ is the probability that $x_r \in S^*$). We have that:
\begin{align*}
p_{\hit} \geq \frac{v(S^*) - v(S)}{\sum_{x \in X}\Big(v(S \cup \{x\}) - v(S)\Big)} \text{.}
\end{align*}
\end{lemma}
\begin{proof}
\begin{align*}
p_{\hit} & = \frac{\sum_{x \in S^*}\Big(v(S \cup \{x\}) - v(S)\Big)}{\sum_{x \in X}\Big(v(S \cup \{x\}) - v(S)\Big)} & \\
             & \geq \frac{\Big(v(S \cup \{x_1\}) - v(S)\Big) + \Big(v(S \cup \{x_1, x_2\}) - v(S \cup \{x_1\})\Big) + \ldots}{\sum_{x \in X}\Big(v(S \cup \{x\}) - v(S)\Big)} & \text{submodularity} \\
             & = \frac{v(S \cup S^*) - v(S)}{\sum_{x \in X}\Big(v(S \cup \{x\}) - v(S)\Big)}  & \\
             & \geq \frac{v(S^*) - v(S)}{\sum_{x \in X}\Big(v(S \cup \{x\}) - v(S)\Big)} & \text{non-decreasing}
\end{align*}
\end{proof}

\begin{theorem}\label{thm:pSeparableMin}
For each non-negative, non-decreasing, submodular, $p$-subseparable set function $v: 2^X \to \reals$ and for each $0 \leq \alpha < 1$, Algorithm~\ref{alg:pSeparableMin} outputs an $\alpha$-approximate solution for the minimization variant of \textsc{BestKSubset}, with probability $(1 - \epsilon)$. The time complexity of the algorithm is $\poly(n,m) \cdot\left\lceil (-\ln \epsilon) / \left(\frac{\alpha - 1}{p\alpha}\right)^K\right\rceil$.
\end{theorem}
\begin{proof}
Let $I$ be an instance of the \textsc{BestKSubset} problem with $v: 2^X \to \reals$ being a non-negative, submodular, $p$-subseparable function. Let $\alpha$, $\alpha > 1$, and $\epsilon$, $0 < \epsilon < 1$ be the parameters of Algorithm~\ref{alg:pSeparableMin}. Let $S^*$ be some optimal solution for $I$.

Let us consider a single call to \texttt{SingleRun} from function \texttt{Main}.
Let $p_{s}$ denote the probability that such a single invocation of function \texttt{SingleRun} returns an $\alpha$-approximate solution. We will prove the lower-bound of $\Big(\frac{\alpha - 1}{p\alpha}\Big)^K$ for the value of $p_s$. Let $\Ev$ denote the event that during such an invocation, at the beginning of some iteration of the ``for'' loop within the function \texttt{SingleRun}, it holds that:
\begin{align}\label{eq:eventCondition}
v(X) - v(S) \leq \alpha \Big(v(X) - v(S^*)\Big) \textrm{.}
\end{align}
If this event occurs, then there is an iteration where we obtain an $\alpha$-approximate solution and the value of the solution can only increase in further iterations. Thus, in such a case \texttt{SingleRun} definitely returns an $\alpha$-approximate solution. A condition reverse to Inequality~\ref{eq:eventCondition} ($v(X) - v(S) < \alpha \Big(v(X) - v(S^*)\Big)$ can be equivalently formulated as follows:
\begin{align}\label{eq:eventCond}
\frac{v(S^*) - v(S)}{v(X) - v(S)} > \frac{\alpha - 1}{\alpha} \textrm{.}
\end{align}

Now, let us consider a single iteration of the ``for'' loop within the function \texttt{SingleRun}, and let us use the notation from Lemma~\ref{lem:algorithm2}. If $S$ (the value of the partial solution at the beginning of the considered iteration) satisfies Inequality~\eqref{eq:eventCond}, then we have:
\begin{align*}
p_{\hit} & \geq \frac{v(S^*) - v(S)}{\sum_{x \in X}\Big(v(S \cup \{x\}) - v(S)\Big)} & \text{Lemma~\ref{lem:algorithm2}} \\
             & \geq \frac{v(S^*) - v(S)}{p\Big(v(X) - v(S)\Big)} & \hspace{-2cm} \text{$p$-subseparability} \\
             & \geq \frac{\alpha - 1}{p\alpha} \textrm{.} & \text{Inequality~\eqref{eq:eventCond}}
\end{align*}

In each iteration of an invocation of the \texttt{SingleRun} function one of the two things can happen. The even $\Ev$ can occur---in such a case we are sure that the function will return an $\alpha$-approximate solution---or it may not occur, and so at the beginning of the next iteration Inequality~\ref{eq:eventCondition} would hold. In the latter case, by our previous reasoning we know that the probability of selecting an element from $S^*$ in this next iteration is at least equal to $\frac{\alpha - 1}{p\alpha}$. Clearly, if in each iteration an element from $S^*$ is selected, then we eventually get an optimal solution, which in particular is an $\alpha$-approximate solution. Consequently, we can lower-bound $p_{s}$ by $\left(\frac{\alpha - 1}{p\alpha}\right)^K$.

To conclude, we use the standard argument that if we make $x = \lceil\frac{-\ln\epsilon}{p_s}\rceil$ independent calls to \texttt{SingleRun}, then the best output from these calls is an $\alpha$-approximate solution with probability at least equal to:
\begin{align*}
1 - \big(1 - p_s\big)^x \geq 1 - e^{\ln\epsilon} = 1 - \epsilon \textrm{.}
\end{align*}
This completes the proof.
\end{proof}

Interestingly, we can slightly modify the proof of Theorem~\ref{thm:pSeparableMin} so that it would apply with the parameter $\frac{\sum_{x \in X}v(\{x\})}{v(X)}$ (for this parameter we do not use the notion of $p$-separability). On the other hand, for this parameter we give weaker approximation guarantees, by approximating the minimization-or-maximization instead of the minimization variant of the problem.  

\begin{theorem}\label{thm:minOrMax}
For each non-negative, non-decreasing and submodular set function $v: 2^X \to \reals$ there exists a randomized FPT approximation scheme for the minimization-or-maximization variant of \textsc{BestKSubset} problem with the parameter $(K, \frac{\sum_{x \in X}v(\{x\})}{v(X)})$.
\end{theorem}
\begin{proof}
Let us fix $\alpha$, $\alpha > 1$, the required approximation ratio. Let $p = \frac{\alpha}{\alpha-1} \cdot \frac{\sum_{x \in X}v(\{x\})}{v(X)}$. We will show that Algorithm~\ref{alg:pSeparableMin} with such value of the parameter $p$ (this parameter is used to determine the number of iterations of the algorithm) is an $\alpha$-approximation algorithm for the minimization-or-maximization variant of the problem.  We repeat the reasoning from the proof of Theorem~\ref{thm:pSeparableMin}, with the following small modification. In the proof of Theorem~\ref{thm:pSeparableMin} we defined $\Ev$ to denote the event that during a single invocation of the \texttt{SingleRun} function from Algorithm~\ref{alg:pSeparableMin}, at the beginning of some iteration of the ``for'' loop, it holds that: $v(X) - v(S) \leq \alpha \Big(v(X) - v(S^*)\Big)$. In this proof we modify this definition saying that $\Ev$ denotes the event when at the beginning of some iteration of the ``for'' loop within the function \texttt{SingleRun}, at least one the following two conditions hold:
\begin{align*}
v(X) - v(S) &\leq \alpha \Big(v(X) - v(S^*)\Big)  \textrm{,} \\
v(S) &\geq \frac{1}{\alpha} v(S^*) \textrm{.}
\end{align*}
Naturally, if this event occurs, then \texttt{SingleRun} definitely returns an $\alpha$-approximate solution for the minimization-or-maximization variant of the problem.
Similarly as in the proof of Theorem~\ref{thm:pSeparableMin}, we will assess $p_{\hit}$, the probability that in a given iteration of the ``for'' loop, the algorithm selects an element from $S^*$ and adds it to the partial solution $S$, under the condition that $\Ev$ has not yet occurred. Assuming that $\Ev$ does not hold, we have that $v(X) - v(S) > \alpha \Big(v(X) - v(S^*)\Big)$ and $v(S) < \frac{1}{\alpha} v(S^*)$. From the first condition we have that $v(X) > \alpha \Big(v(X) - v(S^*)\Big)$ and so $v(X)(\alpha-1) < \alpha(S^*)$.
\begin{align*}
p_{\hit} & \geq \frac{v(S^*) - v(S)}{\sum_{x \in X}\Big(v(S \cup \{x\}) - v(S)\Big)} & \text{Lemma~\ref{lem:algorithm2}} \\
             & \geq \frac{v(S^*) - v(S)}{\sum_{x \in X}\Big(v(\{x\}) - v(\emptyset)\Big)} & \text{submodularity} \\
             & \geq \frac{v(S^*) - v(S)}{\sum_{x \in X}v(\{x\})} = \frac{1}{p} \cdot \frac{\alpha}{\alpha-1} \cdot \frac{v(S^*) - v(S)}{v(X)} &\\
             & \geq \frac{1}{p} \cdot \frac{\alpha}{\alpha-1} \cdot \frac{v(S^*) - \frac{1}{\alpha}v(S^*)}{v(X)} & \text{$\Ev$ has not occurred} \\
             & = \frac{1}{p} \cdot \frac{v(S^*)}{v(X)} \geq \frac{1}{p} \cdot \frac{\alpha-1}{\alpha} \text{.} & \text{$\Ev$ has not occurred}
\end{align*}
We get exactly the same estimation as in the proof of Theorem~\ref{thm:pSeparableMin}. Thus,
with these modifications the proof of Theorem~\ref{thm:pSeparableMin} can be used in this case.
\end{proof}

Algorithm~\ref{alg:pSeparableMin} can be applied to yet another variant of the problem. Let \textsc{BestSubset} be defined similarly to \textsc{BestKSubset}, with the following difference. In \textsc{BestSubset} we are not putting any constraints on the size of the solution, but we rather look for the smallest possible set $S$ such that $v(S) = v(X)$. Interestingly, Algorithm~\ref{alg:pSeparableMin} can be used to find \emph{exact} solutions to \textsc{BestSubset} for non-negative, non-decreasing, submodular, $p$-subseparable set functions, and it will run in FPT time for the parameter $(K, p)$.

\begin{theorem}
For each non-negative, non-decreasing, submodular, $p$-subseparable set function $v: 2^X \to \reals$, the algorithm that runs Algorithm~\ref{alg:pSeparableMin} for consecutive values of the parameter $K$ until it finds a solution $S$, such that $v(S) = v(X)$, is a randomized FPT exact algorithm for the \textsc{BestSubset} problem for the parameter $(K, p)$.
\end{theorem}
\begin{proof}
Let $S^*$ be an optimal solution for the problem. Let us consider a single iteration of the ``for'' loop within the function \texttt{SingleRun}, and let $S$ be the value of the partial solution at the beginning of this iteration. We define $p_{\hit}$ as the probability that in this iteration, an element from $S^*$ is added to the partial solution. We can use a very similar estimation as in the proof of Theorem~\ref{thm:pSeparableMin}: 
\begin{align*}
p_{\hit} & \geq \frac{v(S^*) - v(S)}{\sum_{x \in X}\Big(v(S \cup \{x\}) - v(S)\Big)} & \text{Lemma~\ref{lem:algorithm2}} \\
             & \geq \frac{v(S^*) - v(S)}{p\Big(v(X) - v(S)\Big)} & \hspace{-2.5cm} \text{$p$-subseparability} \\
             & = \frac{1}{p} \textrm{.} & \hspace{-1cm} \text{since}\;\; v(S^*) = v(X) \\
\end{align*}
Thus, we can lower bound the probability that $S^*$ will be found by a single run of the function \texttt{SingleRun}, by $\frac{1}{p^K}$. If we invoke \texttt{SingleRun} a sufficient number of times, then the probability of not finding an optimal solution can be decimated. 
\end{proof}

\section{Applications of the Algorithms}\label{sec:applications}

In this section we show that the assumption about $p$-separability of submodular set functions is plausible. We provide several examples of known computational problems that can be expressed as maximization of $p$-separable, submodular functions. Consequently, we show that the algorithms from Section~\ref{sec:algorithms} are applicable to these problems.

\subsection{Item Selection in Multi-Agent Systems}

Let $N = \{1, 2, \dots n\}$ be the set of agents and let $C = \{a_1, a_2, \ldots, a_m\}$ be the set of \emph{items}. Each agent $i \in N$ is endowed with a \emph{utility function} $u_i: C \to \reals$ that measures how much $i$ desires each of the items. Our goal is to select $K$ items, called \emph{winners}, that in some sense would make the agents most satisfied.
An ordered weighted average (OWA) vector $\lambda$ is a vector of $K$ reals, $\lambda = \langle \lambda_1, \ldots, \lambda_K \rangle$. Given an OWA vector $\lambda$, for each agent $i$ and for each set of $K$ items $S$, we define $u_i(S)$, the satisfaction of $i$ from $S$, in the following way. Let $z_{i, j}(S)$ denote the $j$-th most preferred, from the perspective of agent $i$, element of $S$, that is $z_{i,1}(S), z_{i,2}(S), \ldots z_{i,K}(S)$ are the utilities from $\{u_i(x) : x \in S\}$ sorted in the descending order. Then $u_i(S) = \sum_{j = 1}^K \lambda_j z_j(S)$. The satisfaction of all agents from $S$ is defined as the sum of satisfactions of all the individuals from $S$. For a fixed OWA vector $\lambda$, the problem is, for a given profile of utility functions and for a given integer $K$, to select a subset of $K$ items $S$ that maximizes the total satisfaction of the agents from $S$.

This model captures various natural problems, from winner determination in multiwinner election systems, through recommendation systems, to location problems. For instance the problem of selecting $K$ items under the OWA vector $\lambda = \langle 1, 0, \ldots, 0\rangle$ boils down to the problem of winner determination under Chamberlin and Courant rule~\cite{ccElection}, or to the facility location problem~\cite{far-hek:b:facility-location}. The problem for $\lambda = \langle 1, 1/2, \ldots, 1/K\rangle$ is equivalent to winner determination in the Proportional Approval Voting (PAV) system~\cite{pavVoting}. For more examples of applications of this generic model we refer the reader to the original work of Skowron~et~al.~\cite{owaWinner}.

We say that the agents have $k$-approval utilities if each agent assigns utility equal to $1$ to exactly $k$ items, and utility equal to $0$ to the remaining ones. Such $k$-approval utilities are very popular in the context of social choice, in particular in case of multi-winner election rules~\cite{pavVoting, bau-erd-hem-hem-rot:b:computational-apects-of-approval-voting, azi-bri-con-elk-fre-wal:c:justified-representation, conf/aaai/SkowronF15, conf/atal/AzizGGMMW15}.

Even if the agents have $k$-approval utilities and even for the simple OWA vector $\lambda = \langle 1, 0, \ldots, 0\rangle$ the problem of selecting $K$ items is $\np$-hard~\cite{complexityProportionalRepr} and hard from the perspective of parameterized complexity theory~\cite{fullyProportionalRepr}. Also, it is known that even under these simplifying assumptions, no polynomial-time algorithm can approximate the problem with a better ratio than $1 - \nicefrac{1}{e}$. Thus, this is interesting to see that Algorithms~\ref{alg:submod}~and~\ref{alg:pSeparableMin} are applicable to the problem of selecting $K$ items in the generic OWA-based framework with $k$-approval utilities.

\begin{lemma}\label{lemma:owa}
For a non-decreasing OWA vector and for $k$-approval utilities, the problem of selecting $K$ best items can be expressed as the maximization of a non-negative, nondecreasing submodular function which is (i) $k$-superseparable, and (ii) $k$-subseparable.
\end{lemma}
\begin{proof}
We start from defining an appropriate set function $v$. It is easy to extend the definition of the utility of an agent $i$ from a set of items $S$, to be applicable to sets with a different number of elements than $K$: $u_i(S) = \sum_{j = 1}^{\min(K, |S|)} \lambda_j z_{i, j}(S)$. For each set $S \subseteq C$, we define $v(S)$ as the total satisfaction of the agents from $S$: $v(S) = \sum_{i \in N} u_i(S)$.

Skowron~et~al.~\cite{conf/aaai/SkowronF15} showed that such defined function is submodular. We will show that it is also $k$-superseparable and $k$-subseparable. Since the sum of $p$-superseparable (respectively, $p$-subseparable) set functions is also $p$-superseparable (respectively, $p$-subseparable), it suffices to show that our hypothesis is true for an arbitrary $k$-approval utility function of a single agent $i$, $u_i$.

We fix $S \subseteq C$: let $\ell$ denote the number of elements in $S$ which are approved of by $i$. Since, the utilities of the agents are $k$-approval, there are only $(k-\ell)$ elements in $C\setminus S$ for which $u_{i}(\{x\}) > 0$. For each such an element $x$, we have $\big(v(S \cup \{x\}) - v(S)\big) = \lambda_{\ell + 1}$. Thus:
\begin{align*}
\sum_{x \in X} \big(u_i(S \cup \{x\}) - u_i(S)\big) = (k - \ell) \lambda_{\ell + 1} \textrm{.}
\end{align*}
If $\ell \geq 1$, then:
\begin{align*}
\sum_{x \in X} \big(u_i(S \cup \{x\}) - u_i(S)\big) = (k - \ell) \lambda_{\ell + 1} \geq 0 = k \lambda_1 - k \lambda_1 \geq \sum_{x \in X}u_i(\{x\}) - ku_i(S) \textrm{.}
\end{align*}
If $\ell = 0$, then:
\begin{align*}
\sum_{x \in X} \big(u_i(S \cup \{x\}) - u_i(S)\big) = k \lambda_{1} = \sum_{x \in X}u_i(\{x\}) \geq \sum_{x \in X}u_i(\{x\}) - ku_i(S) \textrm{.}
\end{align*}
Which shows that $u_i$ is $k$-superseparable, and thus that $v$ is $k$-superseparable.
Further, 
\begin{align*}
\sum_{x \in X} \big(u_i(S \cup \{x\}) - u_i(S)\big) &= (k - \ell) \lambda_{\ell + 1} \leq (k - \ell) \big(\sum_j \lambda_j - \sum_{j \leq \ell}\lambda_j\big) \\
                                                    &= (k - \ell)\big(u_i(X) - u_i(S)\big) \leq k\big(u_i(X) - u_i(S)\big) \textrm{.}
\end{align*}
Which shows that $u_i$ is $k$-subseparable. This completes the proof.
\end{proof}

As the corollary we get that Algorithms~\ref{alg:submod}~and~\ref{alg:pSeparableMin} are applicable to the problem.

\begin{corollary}\label{cor:owa}
For each non-increasing OWA vector there exists FPT approximation schemes for the maximization and minimization variants of the problem of selecting $K$ items with $k$-approval utilities for the parameter $(K, k)$.
\end{corollary}

There are numerous consequences of Corollary~\ref{cor:owa}. First, Algorithms~\ref{alg:submod}~and~\ref{alg:pSeparableMin} are applicable to the facility location problem~\cite{Jain:2001:AAM:375827.375845}. Further, since the item selection problem is a model for multi-winner election systems (with candidates corresponding to items and voters to agents), Algorithms~\ref{alg:submod}~and~\ref{alg:pSeparableMin} are applicable to the problem of finding winners under Chamberlin and Courant~\cite{ccElection} and Proportional Approval Voting~\cite{pavVoting} election systems. In all these cases the assumption that the number of approved items is small is realistic. For instance, in case of elections, in some countries, the voting procedure imposes constraints on how many candidates a voter can approve of (for instance, in the Polish parliamentary elections these are only three candidates).
Our results can be also extended to cover geometric utilities, that is the case where the set of utilities of each agent has the form $\{d^{m-1}, d^{m-2}, \ldots, 1 \}$, for some $d > 1$.

\subsection{Matching and Assignment Problems}

The algorithm for $p$-superseparable set functions, i.e., Algorithm~\ref{alg:submod}, is also applicable to variants of assignment problems in bipartite graphs. We provide an example by showing how to apply our results to the \textsc{Weighted-B-$K$-Matching} problem, which is similar in spirit to the item selection problem from the previous subsection. Here, however, we introduce additional capacity constraints, which says that items cannot be shared among too many agents.

In the \textsc{Weighted-B-$K$-Matching} problem we are given a set of vertices $X \cup Y$, a set of edges $E$ (there are no edges neither between the vertices from $X$ nor between the vertices from $Y$), a weight function $w: E \to \reals$, and a capacity function $c: X \to \integers$. The goal is to find a subset of edges with the maximal total weight, such that each vertex $x \in X$ belongs to at most $c(x)$ of the selected edges, each vertex $y \in Y$ belongs to at most one of the selected edges, and altogether there are at most $K$ vertices from $X$ which belong to some 
of the selected edges.

Let us explain the relation between the \textsc{Weighted-B-$K$-Matching} problem and the item selection problem from the previous subsection. We observe that vertices from $X$ can be thought of as items and the vertices from $Y$ can be identified with agents. A weight $w$ of an edge $(x, y)$, $x \in X$ and $y \in Y$, corresponds to the utility that the agent $y$ assigns to the item $x$. In the  \textsc{Weighted-B-$K$-Matching} problem we need to select $K$ items from $X$ and assign each agents to the selected items so that the number of agents assigned to a single selected item does not exceed its capacity. 

This matching problem is at least as hard as the item selection problem from the previous subsection for the fixed OWA vector $\lambda = \langle 1, 0, \ldots, 0\rangle$. In particular, it is $\np$-hard,  for many natural parameters it is hard from the point of view of the parameterized complexity theory, and there exists no polynomial-time algorithm that would approximate it with a better ratio than $1 - \nicefrac{1}{e}$. Yet, our algorithms are also applicable to this problem.

\begin{lemma}\label{lemma:bMatching}
The \textsc{Weighted-B-$K$-Matching} problem with the degree of vertices from $Y$ bounded by $p$ can be expressed as the maximization of a non-negative, nondecreasing submodular function which is $p$-superseparable.
\end{lemma}
\begin{proof}
For each set $S \subseteq X$ we define $v(S)$ as the maximum weight of a matching for the graph induced by the set of vertices $S \cup Y$. Such defined $v$ is non-negative and submodular (\cite{sko-fal-sli:j:multiwinner}, Theorem 9), and can be evaluated in polynomial time.

First, we show that $v$ is $p$-superseparable. For each $S$, let $\mathcal{M}$ denote some optimal matching for the graph induced by the set of vertices $S \cup Y$. For each $S$, and for each $x \in X$, $\big(v(S) + v(\{x\}) - v(S \cup \{x\}) \big)$ is no greater than the weight of such edges $(x', y') \in \mathcal{M}$ that there exists an edge between $x$ and $y'$. This holds because we can construct a feasible matching for the graph induced by $S \cup \{x\} \cup Y$ in the following way. Let us denote the matchings for the graphs induced by $S \cup Y$ and by $\{x\} \cup Y$, as $\mathcal{M}_1$ and $\mathcal{M}_2$ respectively. If we merge $\mathcal{M}_1$ and $\mathcal{M}_2$ and drop each edge $(x', y')$ from $\mathcal{M}_1$ that connects a vertex $y' \in Y$ which is already connected in $\mathcal{M}_2$ (i.e., each edge $(x', y')$ from $\mathcal{M}_1$ such that there exists an edge $(x, y')$ in $\mathcal{M}_2$), we get a feasible matching for the graph induced by $S \cup \{x\} \cup Y$. Consequently, $\sum_{x \in X} \Big(v(S) + v(\{x\}) - v(S \cup \{x\}) \Big)$ is no greater than the weight of $\mathcal{M}_1$ times the bound on the degree of vertices from $Y$. That is:
\begin{align*}
\sum_{x \in X} \Big(v(S) + v(\{x\}) - v(S \cup \{x\}) \Big) \leq p\cdot v(S) \textrm{,}
\end{align*}
which is equivalent to the condition for $p$-superseparability. This proves the thesis.
\end{proof}

\begin{corollary}\label{cor:bMatching}
There exists an FPT approximation scheme for the maximization variant of the \textsc{Weighted-B-$K$-Matching} for the parameter $p$, the bound on the degree of vertices from $Y$.
\end{corollary}

Since finding winners under the Monroe election system~\cite{monroeElection} is computationally equivalent to solving \textsc{Weighted-B-$K$-Matching} with specific weights and specific capacities~\cite{sko-fal-sli:j:multiwinner}, Corollary~\ref{cor:bMatching} implies that for winner determination under the Monroe election system with $k$-approval utilities, there exists an FPT approximation scheme for the parameter $(k, K)$.

\subsection{The \textsc{MaxWeightCover} Problem}\label{sec:maxWeightCover}

In this subsection we explain that all our algorithms are applicable to \textsc{MaxWeightCover}, a generalized variant of the \textsc{MaxCover} problem.

Let us recall that in the \textsc{MaxWeightCover} problem, we are given a set $N = \{e_1, e_2, \ldots e_n\}$ of $n$ elements and a collection $X = \{S_1, \ldots, S_m\}$ of $m$ subsets of $N$. Each element $e_i$ has a weight $w_i$. The goal is to find a subcollection $\calC$ of $X$ of size at most $K$ that maximizes the total weight of covered elements: $\sum_{i: i \in \bigcup \calC}w_i$.
Let us also recall that a \emph{frequency} of an element $e_i$ is the number of sets that contain $e_i$. Frequency of elements is a natural parameter considered in the context of approximability of covering problems~\cite{Vazirani:2001:AA:500776}. To the best of our knowledge, for polynomial-time algorithms, there exists no better guarantee for the \textsc{MaxCover} problem with bounded frequencies of elements than $(1 - \nicefrac{1}{e})$. This is specifically interesting, since such better approximation algorithms exists for the very similar problem \textsc{SetCover}~\cite{Vazirani:2001:AA:500776}. 

In Example~\ref{ex:maxCover} we have shown that the \textsc{MaxWeightCover} problem with the frequency of elements upper-bounded by $p$ can be expressed as the maximization of a non-negative, non-decreasing submodular function which is (i) $p$-superseparable, and (ii) $p$-subseparable, and that the \textsc{MaxWeightCover} problem with the frequency of elements lower-bounded by $p$ can be expressed as the maximization of a non-negative, non-decreasing submodular function which is rev-$p$-subseparable. As a corollary of these observations we get that all our algorithms can be applied to certain natural variants of the \textsc{MaxWeightCover} problem.

\begin{corollary}\label{cor:cover1}
There exists an FPT approximation scheme both for the maximization and for the minimization variant of the \textsc{MaxWeightCover} for the parameter $(K, p)$, where $p$ is the upper-bound on the frequency of the elements.
\end{corollary}

As an implication from Corollary~\ref{cor:ptas}, we also get that:

\begin{corollary}\label{cor:cover2}
Let $\gamma$ be a constant and let $\gamma$-\textsc{MaxWeightCover} be the variant of the \textsc{MaxWeightCover} problem with additional assumption that each element is covered by at least $\gamma$ fraction of the sets from $X$. There exists a PTAS for $\gamma$-\textsc{MaxWeightCover}.
\end{corollary}

Corollaries~\ref{cor:cover1}~and~\ref{cor:cover2} extend the recent results for the \textsc{MaxCover} problem~\cite{conf/aaai/SkowronF15}. Further, interestingly Theorem~\ref{thm:minOrMax} says that there exists a randomized FPT approximation scheme for the minimization-or-maximization variant of the \textsc{MaxWeightCover} problem, for the parameter $(K, p_{av})$, where $p_{av}$ is an average frequency of an element.

\section{Conclusions}

In this paper we have considered approximation algorithms for the problem of maximizing submodular set functions. Since it is known that the greedy algorithm achieves approximation ratio of $(1 - \nicefrac{1}{e})$ and that no polynomial-time algorithm can approximate the problem better, we have initiated the study on algorithms for the problem that run in super-polynomial time. In our study we have followed the approach of parameterized complexity theory. We have identified three new properties of set functions, called $p$-separability properties, and we have considered the parameter $p$ from the definition of $p$-separability. For $p$ combined with $K$, the size of the solution, we have shown that the problem can be arbitrarily well approximated by algorithms that run in FPT time. We have justified our choice of the parameters by providing numerous examples of computational problems which can be expressed as maximization of submodular $p$-separable set functions, and by arguing that in many real-life problems the value of $p$ is small.

We have exposed the difference between approximating the maximization and minimization variants of the problem and we have provided a new weaker definition of approximability: approximation of minimization-or-maximization variant of the problem. We have shown that it is possible to approximate our problem arbitrarily well using algorithms that run in FPT time for the more general parameter $\Big(K, \frac{\sum_{x \in X}v(\{x\})}{v(X)}\Big)$.

There are many natural ways in which this research can be extended. We believe that one of the promising approaches is to consider the problem with additional constraints, such as knapsack constraints or matroid constraints.

\subsection*{Acknowledgments.}
The author thanks Piotr Faliszewski for his helpful comments.
This research has been supported by Europe Research Grant ERC-StG 639945 and by a Humboldt Research Fellowship for Postdoctoral Researchers.

\bibliographystyle{abbrv}
\bibliography{main}

\end{document}